\documentclass[runningheads]{llncs}
\usepackage{graphicx}
\usepackage{ifthen}
\usepackage{xspace}
\usepackage{algorithm}
\usepackage{algorithmic}
\usepackage{amsmath, amssymb, amsfonts}
\usepackage{fullpage}
\usepackage{xcolor}

\begin{document}

\title{An Alternating Algorithm for Finding Linear Arrow-Debreu Market Equilibria\thanks{This paper appeared in preliminary form in ICCOPT 2019 (the 6th International Conference on Continuous Optimization).}}

\author{Po-An Chen\inst{1}\thanks{Corresponding author, supported in part by MST 105-2221-E-009-104-MY3} \and Chi-Jen Lu\inst{2}\and Yu-Sin Lu\inst{3}}
\authorrunning{P.-C. Chen, C.-J. Lu, and Y.-S. Lu}
\institute{Institute of Information Management, National Chiao Tung University\\1001 University Rd, Hsinchu City 300, Taiwan \\\email{poanchen@nctu.edu.tw}\and
Institute of Information Science, Academia Sinica\\128 Academia Road, Section 2, Nankang, Taipei 11529, Taiwan
\\\email{cjlu@iis.sinica.edu.tw}\and
Industrial Technology Research Institute\\195, Section 4, Chung Hsing Rd, Chutung, Hsinchu County, Taiwan \\\email{s8700937@gmail.com}}

\maketitle
\begin{abstract}  % put your abstract here!
Motivated by the convergence result of mirror-descent algorithms to market equilibria in linear Fisher markets, it is natural for one to consider designing dynamics (specifically, iterative algorithms) for agents to arrive at linear Arrow-Debreu market equilibria.
Jain \cite{jain} reduced equilibrium computation in linear Arrow-Debreu markets to the equilibrium computation in bijective markets,
where everyone is a seller of only one good and a buyer for a bundle of goods.
In this paper, we design an algorithm for computing linear bijective market equilibrium, based on solving the rational convex program formulated by Devanur et al. The algorithm repeatedly alternates between a step of gradient-descent-like updates and a distributed step of optimization exploiting the property of such convex program.
Convergence can be ensured by a new analysis different from the analysis for linear Fisher market equilibria.
\end{abstract}

\section{Introduction}
A market can be thought of as an algorithm or mechanism that implements a social choice of redistribution of various goods between agents (as buyers and sellers for goods) via pricing.
In mathematical economics, \emph{exchange market} models were first proposed by Walras in 1874 and later by Arrow and Debreu in 1954 along with the concept of \emph{market equilibrium} \cite{arrow}.
These exchange market models are used to capture the essence of complicated real-world markets.

In the Arrow-Debreu market model, each agent has an initial endowment of divisible goods and a utility function for purchasing a bundle of goods that maximizes her utility when every agent uses the revenue from selling her initial endowment.
The Fisher market model \cite{birnbaum} can be seen as a special case of the Arrow-Debreu model. In it, each buyer is subject to her budget constraint instead of the revenue from selling her initial endowment.
There is a market equilibrium\footnote{Note that the concept of market equilibria is \emph{not} the same as the concept of Nash equilibria. A market equilibrium does not allow buyers to strategically report their interests in different goods in order to maximize their own utilities.} if two conditions are satisfied:
\begin{itemize}
\item \emph{Buyer Optimality}, in which every agent uses the revenue from selling her initial endowment to purchase a bundle of goods to maximize her own utility
\item \emph{Market Clearance}, in which the total demand for every good equals the total supply
\end{itemize}
The celebrated theorem by Arrow and Debreu \cite{arrow} proved the existence of a market equilibrium under some mild necessary conditions for the utility functions.
The most common types of utility functions from this class, which are described in more detail in the section of preliminaries, are the \emph{linear}, \emph{Leontief}, and \emph{Cobb-Douglas} functions, and belong to the important class of \emph{Constant Elasticity of Substitution (CES)} functions \cite{solow,arrow:chenery}.

%\subsection{Our Results}
Formally, an Arrow-Debreu market $M$ consists of a set $A = \lbrace 1,...,n \rbrace$ of agents that trade a set $G=\lbrace 1,...,m \rbrace$ of divisible goods among themselves.
%(We we may want to work on a graph that is not necessarily complete.)
Each unit of good $g \in \lbrace 1,...,m \rbrace$ can be bought by an agent at price $p_g$. The vector of prices is $\mathbf{p} \in \mathbb{R}^m_+$ ($\mathbb{R}_+=\lbrace x \geq 0 \rbrace$).
Each agent $i$ purchases a consumption plan $\mathbf{x}_i \in \mathbb{R}^m_+$.
Each agent $i$ has an initial endowment $\mathbf{w}_i \in \mathbb{R}^m_+$ of the $m$ goods, where $w_{ij}$ is the amount of good $j$ initially held by $i$. These initial goods can be sold to other agents and thus provide agent $i$ with revenue. The revenue can in turn be used to purchase other goods.
Every agent $i$ has a utility function $U_i$ : $\mathbb{R}^m_+ \rightarrow \mathbb{R}_+$, where $U_i(\mathbf{x}_i)$  describes how much utility agent $i$ receives from the consumption plan $\mathbf{x}_i$.
An Arrow-Debreu market can then be described as $M=(A,G,u,w)$.When $G=A$, i.e., $w_{ij}=1$ if $i=j$ and $0$ otherwise, the Arrow-Debreu market becomes a bijective market.
The utility of agent $i$ for the good of agent $j$ is $u_{ij} \geq 0$.
The directed graph $(A,E)$ has an edge $ij$ for every pair with $u_{ij} > 0$ where set $E$ consists of all the edges.
Jain \cite{jain} reduced Arrow-Debreu markets to bijective markets as follows.
If a good is included in the initial endowment of multiple agents, we give a different name to each good.
If an agent has $k$ goods in the endowment, we split the agent into $k$ copies that have the same utility function and each owns one of the goods.
Consequently, we can focus on the equilibrium computation for bijective markets since any solutions would also work for Arrow-Debreu markets.

Given the convergence of distributed generalized gradient-descent algorithms (i.e., mirror-descent) to market equilibria in the Fisher and CCVF markets \cite{birnbaum,cheung},
it is natural for one to consider \emph{designing dynamics (specifically, iterative algorithms) for agents to arrive at market equilibria in Arrow-Debreu markets}.\footnote{We can even require the dynamics to be of the no-regret property to provide incentives for the agents if possible \cite{kleinberg:piliouras:tardos:load,CL14,CL15,CL16}.}
In this paper, instead of using generic algorithms for solving a convex program, i.e., convex program solvers,
we design (partially-distributed) iterative algorithms for solving the following convex program formulated by Devanur et al. in \cite{devanur:garg} for \emph{bijective markets},
where everyone is a seller of only one good and a buyer for a bundle of goods:\footnote{Jain \cite{jain} reduced the equilibrium computation in linear Arrow-Debreu markets to the equilibrium computation in bijective markets.}
%We give the formulation here for the convenience of discussion.
%We will elaborate on such formulation more in the section of preliminaries.
\begin{eqnarray} \label{eq:convex}
&\min \quad \sum_{j} p_j \log\frac{p_j}{\beta_j} - \sum_{i,j}b_{ij}\log u_{ij} \\
& \nonumber \quad\quad\quad \sum_i b_{ij}=p_j \ \forall j \\
& \nonumber \quad\quad\quad \sum_j b_{ij}=p_i \ \forall i \\
& \nonumber \quad\quad\quad\,\, u_{ij}/p_j \leq 1/\beta_i\ \forall i,j \\
& \nonumber \quad\quad\quad\quad p_i \geq 1 \ \forall i \\
& \nonumber \quad\quad\quad\quad b_{ij}\geq 0 \ \forall i,j, \beta_i \geq 0 \ \forall i,
\end{eqnarray}
where $b_{ij}$ is paid by agent $i$ to agent $j$ for good $j$, $p_i$ ``acts like" the endowed budget $B_i$ of agent $i$ in the Fisher model
(yet there is no actual budget but only the initial endowment of goods for agent $i$),
and $\beta_i$ is agent $i$'s inverse of ``best" bang per buck.
Our algorithm for computing linear Arrow-Debreu market equilibria is based on solving this rational convex program that captures buyer optimality and market clearance.
Generic solvers that mostly search through the feasible space for an optimal solution are simply designed for any optimization problem, do not take advantage of the structure of the mathematical program, and may sacrifice scalability, i.e., they may not be fast enough when facing large problems. Since we are dealing with a special convex program characterizing linear bijective market equilibria, update dynamics that exploit the structure of this program may enjoy better scalability and, from a computational perspective, enable potentially (at least partially) distributed implementation rather than entirely centralized computation of generic solvers.
Thus, technically generic solvers may not be a good candidate for comparison with our partially-distributed algorithm.
On the other hand, we argue the necessity of our design and compare our algorithm mainly with block coordinate gradient descent type methods \cite{beck:tetruashvili,luo:tseng} later in Section~1.1.

\subsubsection{Mirror Descents for Fisher Markets.}
Recall that when designing distributed algorithms via mirror descents for Fisher markets \cite{birnbaum},
each $\beta_i$ is endogenously set to $1$ in the convex program above (there is no concept of variable $\beta_i$ so there is no need for to update $\beta_i$).
The convex program is then as follows.
\begin{eqnarray*}
&\min \quad \sum_{j} p_j \log p_j - \sum_{i,j}b_{ij}\log u_{ij} \\
& \nonumber \quad\quad\quad \sum_i b_{ij}=p_j \ \forall j \\
& \nonumber \quad\quad\quad \sum_j b_{ij}=B_i \ \forall i \\
& \nonumber \quad\quad\quad\quad b_{ij}\geq 0 \ \forall i,j.
\end{eqnarray*}
The objective function becomes
\[\varphi(\mathbf{b})=\sum_j p_j\log p_j-\sum_{i,j}b_{ij}\log u_{ij}=\sum_{i,j}b_{ij}\log\frac{p_j}{u_{ij}},\] and the feasible space is
\[\mathcal{S}=\{\mathbf{b}\in\mathbb{R}^{n\times m}:\sum_j b_{ij}=B_i\forall i,b_{ij}\geq 0\forall i,j\}\]
where $p_j=\sum_i b_{ij}$.
The components of the gradient of $\varphi$ are thus
\[(\nabla\varphi(\mathbf{b}))(i,j)=1-\log\frac{u_{ij}}{p_j}.\]
The projection back to the feasible space is applied on vector $\mathbf{b}'' = (\mathbf{b}''_i)_i = (b''_{ij})_{ij}$ after the update such that $\mathbf{b}'' =\mathbf{b}-\eta\nabla\varphi(\mathbf{b})$ for a learning rate $\eta$ and $\mathbf{b}'$ is the feasible point after projection.
The choice of the regularizer and thereby the corresponding Bregman divergence decides the actual distance measure for projection.
The Proportional Response (PR) dynamics \cite{zhang} are the weight updates that result from the Bregman divergence being the KL-divergence by choosing the relative entropy as the regularizer and $\eta=1$:
\[b'_{ij}=\frac{1}{Z_i}b_{ij}(\frac{u_{ij}}{p_j})\]
where $Z_i$ is chosen such that $\sum_j b_{ij}=B_i$.
Due to the property of the PR dynamics, the update of $b_{ij}$ does not need to depend on the information of any other agent (given that any price $p_j$ is publicly known), so it can be performed in a \emph{distributed} way. When the Bregman divergence is defined to be half the square of $l_2$ norm, the algorithm becomes the gradient descent algorithm.

\subsection{Our Results: Alternating Algorithms for Bijective Markets.}
In the above formulation~(\ref{eq:convex}), other than $b_{ij}$, there is one more set of variables that we need to consider for Arrow-Debreu markets: $\beta_i$.
One can simply choose to perform standard gradient updates (with projection) with respect to the concatenation of $(b_{ij})_{ij}$ and $(\beta_i)_i$. However, this approach would preclude any possibility of distributed updates of $(b_{ij})_{ij}$ or/and $(\beta_i)_i$. Since distributed computation is a desirable property to have as in \cite{birnbaum} and \cite{CL14,CL15,CL16}, we would like to utilize the structure of this special convex program to at least let some variables be updated distributedly.
If we separate the updates of $b_{ij}$ and $\beta_i$, we can alternate between updating $b_{ij}$ and updating $\beta_i$.
Fixing one set of variables temporarily, we can update the other variable set to decrease the objective value. This can be done repeatedly.
One natural attempt is to update $b_{ij}$ using mirror-descent types of approaches and updating $\beta_i$ simply optimizing the objective value.
Careless updates of $\beta_i$ would defeat the purpose of optimizing the objective over $b_{ij}$ (in the manner described above).
%This is where the difficulty comes from.

Observe that $\beta_i\leq\frac{p_j}{u_{ij}}$ for all $i,j$ and that the objective is non-increasing with $\beta_i$ (since the derivative with respect to $\beta_i$ is $-\frac{p_i}{\beta_i}\leq 0$). Thus, we know that $\beta_i = \min_j p_j/u_{ij}$ is the best possible value if we fix the value of $b_{ij}$, and this should be the update rule for $\beta_i$, which can be performed distributedly since it does not involve any information from other agents and every price is just a common information.
Note that mirror descents and the convergence analysis in \cite{birnbaum,CL14} will not work directly since we will modify $\beta_i$ before each modification of $b_{ij}$.
Nevertheless, we show that \emph{a convergence analysis mainly for gradient descents with certain condition} will do the job.
Thus, our approach in this paper is to alternate between
\begin{enumerate}
\item a gradient-descent step of updating $b_{ij}$ (with the previously chosen $\beta_{i}$) with projection\footnote{
The projection would be applied to the whole vector $\mathbf{b}$ as in \cite{birnbaum}.
Alternatively, one may apply mirror descents with respect to $b_i$ separately and distributedly, and 
show that, jointly, the objective still converges to minimum (as has been done for congestion games
\cite{CL14,CL15,CL16}). This can be left as future work.}, and
\item an optimization step (instead of the gradient-descent step) of updating $\beta_i$ (with $p_j$ determined by the previously chosen $b_{ij}$ in Step (1)).
\end{enumerate}
Specifically, Step (2) is done by each agent $i$ computing $\beta_i=\min_j p_j/u_{ij}$ (the inverse of the best bang per buck given the previously fixed $p_j$ in Step (1)). This can be done distributedly for each agent $i$.

To prove the convergence, we have to argue the following.
\begin{itemize}
\item Each gradient step decreases the current objective value by some amount when using a new analysis for gradient descents with a convex objective other than the analyses in \cite{birnbaum} and \cite{CL14}.
\item Step (2) does not increase the objective value given the previously chosen $b_{ij}$ in Step (1).
\end{itemize}
Note that we need such a new analysis because the convex function for which we perform the gradient-descent step keeps changing. It is the convex objective in which each $\beta_i$ is fixed to the value chosen in previous iteration.

Our approach bears a resemblance to the well-known block coordinate gradient descent type methods \cite{beck:tetruashvili,luo:tseng}, including the Block Coordinate Gradient Descent/Projection (BCGD and BCGP) algorithm and the alternating minimization algorithm \cite{beck:tetruashvili}. Our $(b_{ij})_{ij}$ and $(\beta_i)_i$ are like blocks in the BCGP algorithm yet there are two main reasons that prevent us from directly using BCGP: first, the BCGP algorithm assumes the global domain of the vector of all variables is a Cartesian product of all blocks' domains (which are independent) and there is no constraint involving more than one block of variables, but our constraints $\frac{u_{ij}}{p_j}\leq\frac{1}{\beta_i}\forall i,j$ violates this; secondly, we cannot perform the update of $(\beta_i)_i$ distributedly if we also do gradient-descent updates for $(\beta_i)_i$.
Thus, we do need to design a mixture of the gradient-descent update (with projection) and optimization update in this paper in order to best make use of the characteristics of the convex program modeling linear bijective market equilibria although our current analysis can only provide a lower convergence rate than the BCGP and two-block alternating minimization algorithms in \cite{beck:tetruashvili}. Nevertheless, looking at the convergence results, our bound $\lambda$ regarding the Hessian plays a similar role as the Lipschitz constant bounds in the BCGP algorithm. We leave the job of giving a better convergence rate to immediate future work.

\subsection{Related Work}
\subsubsection{Connection to block coordinate gradient descent type methods.}
The BCGD/BCGP method \cite{beck:tetruashvili} has each iteration consisting of a gradient step without/with projection with respect to some block of variables taken in a cyclic (or random) order. A global sublinear rate of convergence is established and even accelerated when the problem is unconstrained. Our convergence analysis in this paper shares a high-level similarity with those in decrease of the objective value for every iteration (in terms of some sort of recurrence relation) and then showing the global convergence rate using the recurrence recurrence relation. Since blocks of variables are updated sequentially, their convergence analysis requires the independence of the blocks' domains, which is not the case in our convex program for linear bijective market equilibria. We have constraints that involve two blocks of variables so the domains of blocks are not independent. Also, our algorithm has a learning rate, which is a common concept in machine learning when design algorithms and deriving convergence results, while the BCDG/BCDP algorithms has no such concept.

\subsubsection{Dynamics and computation for market equilibria.}
For Fisher market equilibria, the Eisenberg-Gale convex program \cite{eisenberg} can capture the equilibrium allocation for buyers with utility functions from the same class in the CES family.
The problem of equilibrium computation was introduced to the theoretical computer science community by Devanur et al. \cite{devanur}.
The proportional response dynamic is equivalent to a generalized gradient-descent algorithm with Bregman divergences on a convex program that captures the equilibria of Fisher markets with linear utilities \cite{zhang,birnbaum}.
The tatonnement process is a simple and natural rule for updating prices in markets.
Cheung et al. \cite{cheung} established that tatonnement is like a generalized gradient descent that uses the Bregman divergence for the class of Convex Conservative Vector Field (CCVF) markets.

%A algorithm can compute approximate equilibria in Arrow-Debreu Market by Kakade et al. \cite{kakade2004graphical} in 2004.
Jain \cite{jain} reduced the equilibrium computation in linear Arrow-Debreu markets to the equilibrium computation in bijective markets.
Chen et al. \cite{chen:dai} showed that the problem of computing an Arrow-Debreu market equilibrium with additively separable utilities is PPAD-complete.
Devanur et al. \cite{devanur:garg} formulated a rational convex program for the linear Arrow-Debreu model to characterize a market equilibrium.

\subsubsection{Markets on networks.} In almost all the market models we discussed previously, every buyer can shop for goods from every seller.
This implies that the agents (buyers/sellers) reside in an underlying network that is a complete graph.
In reality, trading may only happen between two immediate neighbors in a graph that is not necessarily complete.
Kakade et al. \cite{kakade} considered Arrow-Debreu markets with such generality and designed algorithms that compute approximate equilibria.

\subsubsection{Market games.} If buyers are allowed to strategically report their interests in different goods in order to maximize their own utilities,
the market becomes a market game \cite{adsul,branzei} in which Nash equilibrium can be used as the equilibrium concept.
The Fisher market game was first studied by Adsul et al. \cite{adsul} for buyers with linear utility functions.
Adsul et al. showed the existence of pure Nash equilibria under mild assumptions and provided the conditions necessary for a pure Nash equilibrium to exist.
Br\^{a}nzei et al. \cite{branzei} showed that a Fisher market game for buyers with linear, Leontief, and Cobb-Douglas utility functions always has a pure Nash equilibrium.
Br\^{a}nzei et al. also bounded the price of anarchy for Fisher market games.

\section{Preliminaries} \label{sec:prelim}

We now introduce the general Constant Elasticity of Substitution (CES) utility function family as follows:
\begin{equation}
U_i(\mathbf{x}_i)=(\displaystyle\sum_{j=1}^m u_{ij}\cdot x_{ij}^\rho)^\frac{1}{\rho}
\end{equation}
where $- \infty < \rho \leq 1$, $\rho \neq 0$. The Leontief, Cobb-Douglas, and linear utility functions shown below are given when $\rho$ approaches $- \infty$, approaches $0$, and equals $1$, respectively.
\begin{equation}
Leontief: \, U_i(\mathbf{x}_i)=\min_{j \in m}\lbrace \frac{x_{ij}}{u_{ij}} \rbrace
\end{equation}
\begin{equation}
Cobb-Douglas: \, U_i(\mathbf{x}_i) = \prod_{j \in m} x_{ij}^{u_{ij}}
\end{equation}
\begin{equation}
Linear: \, U_i(\mathbf{x}_i)= \sum_{j \in m}u_{ij}x_{ij}
\end{equation}
The Leontief function captures the utility of items that are perfect complements.
The Linear function captures the utility of items that are perfect substitutes.
The Cobb-Douglas function represents a perfect balance between complements and substitutes.

In a market equilibrium, we have a set of prices $p : A\rightarrow\mathbb{R}_+$ and allocations
$x : E\rightarrow\mathbb{R}_+$ that satisfy the following conditions \cite{devanur:garg}.
First, agents use an optimal plan at price $\mathbf{p}$ if every agent is allocated a utility-maximizing bundle subject to their constraints.
This reduces to the following.
\begin{itemize}
\item For every $i\in A$, if $x_{ij} > 0$, then $u_{ij}/p_j$ is the maximal value over $j\in A$.
\item $p_i > 0$ for every $i\in A$.
\end{itemize}
The market clears if $\sum_{i\in A}x_{ij}=1$\footnote{W.l.o.g., we assume that $\sum_{i\in A}w_{ij}=1$ for every $j\in A$ since only $w_{ii}=1$ and $w_{ij}=0$ for $i\neq j$ in bijective markets.} for every $j\in A$, i.e., every good is fully sold,
and $p_i=\sum_{j\in A}x_{ij}p_j$ for every $i\in A$, i.e., the money spent by agent $i$ equals her income $p_i$.

\subsection{Convex Program for Linear Bijective Markets}
%In the Devanur's program, the optimum value is $0$, and the prices $p_i$ in an optimal solution give a market equilibrium with allocations $x_{ij} = \frac{b_{ij}}{p_j}$. In fact, the convex program is for linear bijective market model, the bijective market model can transform to Arrow-Debreu market model. The convex program is formulated the following.
Our algorithm for computing linear Arrow-Debreu market equilibria is based on solving the rational convex program formulated by Devanur et al. in \cite{devanur:garg}.
Devanur et al. presented a rational convex program for the linear bijective market that guarantees the existence of a market equilibrium.
%and answered the open questions of Vazirani \cite{vazirani}.
We need to formally introduce the convex program here before giving our alternating algorithm in the next section.

\begin{eqnarray} \label{eq:cp}
&\min \quad \sum_{j} p_j \log\frac{p_j}{\beta_j} - \sum_{ij}b_{ij}\log u_{ij} \\
& \nonumber \quad\quad\quad \sum_i b_{ij}=p_j \ \forall j \in A \\
& \nonumber \quad\quad\quad \sum_j b_{ij}=p_i \ \forall i \in A \\
& \nonumber \quad\quad\quad u_{ij}/p_j \leq 1/\beta_i\ \forall i,j \in E \\
& \nonumber \quad\quad\quad\quad p_i \geq 1 \ \forall i \in A \\
& \nonumber \quad\quad\quad\quad \mathbf{b},\mathbf{\beta} \geq 0
\end{eqnarray}

Note that $u_{ij}$ is the given utility of agent $i$ for the item of agent $j$.
Variable $b_{ij}$ is the money paid by agent $i$ to agent $j$.
Variable $\beta_i$ represents the inverse \emph{best} bang-per-bucks of agent $i$.
Here, variable $p_i$, which is treated as a function of variables $b_{ij}$ for all $j$, corresponds to the value of agent $i$'s initial endowment since $\sum_j b_{ij}=p_i$.
Let $\mathbf{b}=(b_{ij})_{ij}$, $\mathbf{\beta}=(\beta_i)_i$, and $\mathbf{y}=(\mathbf{b},\mathbf{\beta})$.
The money spent by agent $i$ equals $i$'s income $p_i$ (revenue from selling her initial endowment) at market equilibrium, i.e.,
$p_i=\sum_{j\in A}x_{ij}p_j$. The allocation, which is dependent on $b_{ij}$ and $p_j$, can be derived by $x_{ij}=b_{ij}/p_j$.
Therefore, we can focus on $\mathbf{y}$.
The objective function is
\[\Phi(\mathbf{y})= \quad \sum_{j} p_j \log\frac{p_j}{\beta_j} - \sum_{ij}b_{ij}\log u_{ij}.\]
The space of feasible solutions is
\begin{eqnarray*}
\mathcal{S}&=&\{\mathbf{b}\in\mathbb{R}^{n\times m},\mathbf{\beta}\in\mathbb{R}^n:b_{ij}\geq 0\forall i,j,u_{ij}/\sum_i b_{ij} \leq 1/\beta_i\ \forall i,j,\\
&&\sum_j b_{ij}\geq 1 \forall i\}.
\end{eqnarray*}
The space of feasible $\mathbf{b}$'s fixing the values of $\mathbf{\beta}$ is
\begin{eqnarray*}
\mathcal{S}_\mathbf{\beta}&=&\{\mathbf{b}\in\mathbb{R}^{n\times m}:b_{ij}\geq 0\forall i,j,u_{ij}/\sum_i b_{ij} \leq 1/\beta_i\ \forall i,j,\\
&&\sum_j b_{ij}\geq 1 \forall i\}.
\end{eqnarray*}

From \cite{devanur:garg}, the following condition is necessary for the existence of an equilibrium.
(The argument is provided in the appendix.)\\\\
(*) \emph{For every strongly connected component $S\subseteq E$ of the digraph $(A,E)$,
if $|S| = 1$, then there is a loop incident to the node in $S$.}\\

The result from \cite{devanur:garg} shows the feasibility of such a convex program.
\begin{theorem}[\textsc{Thm}~1.1 of \cite{devanur:garg}] \label{thm:existence}
Consider an instance of the linear Arrow-Debreu market given by graph $(A,E)$ and
the utilities $u : E\rightarrow\mathbb{R}_+$. The convex program (\ref{eq:cp}) is feasible if and
only if (*) holds, and, in this case, \textbf{the optimum value is 0}, and the prices $p_i$ in an
optimal solution give a market equilibrium with allocations $x_{ij} = b_{ij}/p_j$.
Further, if all utilities are rational numbers, then there exists a market equilibrium with all prices
and allocations also rational and of a bitsize polynomially bounded in the input size.
\end{theorem}

We aim to design partially-distributed algorithms to solve this rational convex program.

\section{Alternating Algorithm} \label{sec:algo}
Let $\mathcal{K}$ be the convex feasible space with diameter $d$,
and the objective $\phi: \mathcal{K}\rightarrow\mathbb{R}$ is a convex function satisfying the property that there exists positive $\lambda,\gamma\in\mathbb{R}$ such that for any $\xi\in\mathcal{K}$, $0\preceq\nabla^2\phi(\xi)\preceq\lambda I$, and $\|\nabla\phi(\xi)\|_2\leq\gamma$. Note that $d$ and $\gamma$ are bounds that are commonly needed for convergence analysis while how to choose $\lambda$ can be shown in a given problem.

In our case, assume $0\leq\Phi(\mathbf{y})\leq 1$ for any $\mathbf{y}$ w.l.o.g. (by Theorem~\ref{thm:existence}), and define $\phi_\mathbf{\beta}(\mathbf{b})=\Phi(\mathbf{b},\mathbf{\beta})$,
which is the objective function fixing a value of $\beta$;
$\mathcal{K}$ consists of the feasible solutions
\begin{eqnarray*}
\mathcal{S}_{\mathbf{\beta}}&=&\{\mathbf{b}\in\mathbb{R}^{n\times m}:b_{ij}\geq 0\forall i,j,u_{ij}/\sum_i b_{ij} \leq 1/\beta_i\ \forall i,j,\\
&&\sum_j b_{ij}\geq 1 \forall i\},
\end{eqnarray*}
given the values of $\mathbf{\beta}$.
For all $t$, we have that
$\mathcal{S}_{\mathbf{\beta}_t}$ satisfy the diameter bound $d$; $\phi_{\mathbf{\beta}_t}$ satisfy the property with appropriate settings of $\gamma$ and $\lambda$ such that $\gamma$ is an upper bound on each $\|\nabla\phi_{\mathbf{\beta}_t}(\cdot)\|_2$ while
setting $\lambda$ is non-trivial:
specifically, we have the following proposition, whose proof is in Appendix~\ref{pro:hessian}.
\begin{proposition} \label{pro:hessian}
For all $\mathcal{K}_t:=\mathcal{S}_{\mathbf{\beta}_t}$ and $\xi\in\mathcal{K}_t$,
$0\preceq\nabla^2\phi_{\mathbf{\beta}_t}(\xi)\preceq\lambda I$ with $\lambda:=n$.
\end{proposition}

We can obtain the gradient $\nabla\phi_{\mathbf{\beta}_t}(\mathbf{b})$ with respect to $\mathbf{b}$ by taking the derivative of $b_{ij}$ for each $i,j$, i.e.,
\begin{eqnarray}
\nabla\phi_{\mathcal{\beta}_t}(\mathbf{b})=(\frac{\partial\phi_{\mathbf{\beta}_t}(\mathbf{b})}{\partial b_{ij}})_{ij}
=(1-\log\frac{\beta_j u_{ij}}{p_j})_{ij}.
\end{eqnarray}
Let $\mathbf{b}_t\in\mathcal{K}_t$, $\mathbf{b}'=\mathbf{b}_t-\eta\nabla\phi_{\mathbf{\beta}_t}(\mathbf{b}_t)$, and $\mathbf{b}_{t+1}=\Pi_{\mathcal{K}_t}(\mathbf{b}')=\arg\min_{\mathbf{z}\in\mathcal{K}_t}\|\mathbf{z}-\mathbf{b}'\|_2^2$, for a learning rate $\eta\leq\frac{\lambda d^2}{\gamma^2}$, which affects the convergence rate and can be chosen later.
%Denote that $\mathbf{y}'=(\mathbf{b}',\mathbf{\beta}')$ and $\tilde{\mathbf{y}}=(\mathbf{b}',\mathbf{\beta}_t)$.
The idea is to keep the current values of $\mathbf{\beta}_t$ but update $\mathbf{b}$ using the gradient with respect to $\mathbf{b}$.
We denote component $i,j$ of $\mathbf{b}$ at $t$ as $\mathbf{b}_t(i,j)$ and component $i$ of $\mathbf{\beta}$ at $t$ as $\mathbf{\beta}_t(i)$.

\begin{algorithm} \label{eq:alg}
\caption{Alternating algorithm}
\begin{algorithmic}[1]
\STATE $\mathbf{b}_{t+1}=\Pi_{\mathcal{K}_t}(\mathbf{b}_t-\eta\nabla\phi_{\mathbf{\beta}_t}(\mathbf{b}_t))$.
\STATE $\beta_{t+1}(i) = \min_j p_{t+1}(j)/u_{ij}$ for each $i$ where $p_{t+1}(j)=\sum_i b_{t+1}(i,j)$.
\end{algorithmic}
\end{algorithm}
Thus, this is: (1) first moving from point $(\mathbf{b}_t,\mathbf{\beta}_t)$ to point $(\mathbf{b}_{t+1},\mathbf{\beta}_t)$, and (2) then moving from point $(\mathbf{b}_{t+1},\mathbf{\beta}_t)$ to point $(\mathbf{b}_{t+1},\mathbf{\beta}_{t+1})$.
Repeat these two steps, and we need to show that the objective value converges.
That is, we have to show that Step (2) does not increase the objective value given $\mathbf{b}_{t+1}$,
and that the gradient step decreases the current objective value by some amount.
We do this by using a new analysis for gradient descents with a convex objective.
We derive a market equilibrium for a linear Arrow-Debreu market.

\subsection{Convergence}
We introduce the following lemma that is for any immediate objective values from standard gradient-descent updates.
Using this lemma whose technical part will be proved in the next subsection, we then obtain the convergence result for our alternating algorithm, i.e., Theorem~\ref{thm:rate}.
Intuitively, this critical lemma means given a fixed value of $\mathbf{\beta}$, under some condition the gradient descents guarantee that the update of $\mathbf{b}$ decreases the objective value by a certain amount.
The update of $\mathbf{\beta}$ only speeds up the convergence process.
\begin{lemma} \label{thm:convergence}
Let $\textbf{b}_t\in\mathcal{K}_t$, $\textbf{b}'=\textbf{b}_t-\eta\nabla\phi_{\beta_t}(\textbf{b}_t)$, and $\mathbf{b}_{t+1}=\Pi_{\mathcal{K}_t}(\textbf{b}')=\arg\min_{\mathbf{z}\in\mathcal{K}_t}\|\mathbf{z}-\mathbf{b}'\|_2^2$,
for a learning rate $\eta\leq\frac{\lambda d^2}{\gamma^2}$.
Then, whenever $\phi_{\mathbf{\beta}_t}(\mathbf{b}_t)-\phi_{\mathbf{\beta}_t}(\mathbf{q})\geq d\gamma\sqrt{6\lambda\eta}$ for $\textbf{q}_t=\arg\min_{\mathbf{z}\in\mathcal{K}_t}\phi_{\mathbf{\beta}_t}(\mathbf{z})$,
\begin{eqnarray}
\phi_{\mathbf{\beta}_t}(\mathbf{b}_{t+1})\leq\phi_{\mathbf{\beta}_t}(\mathbf{b}_t)-\frac{\eta}{4d^2}(\phi_{\mathbf{\beta}_t}(\mathbf{b}_t)-\phi_{\mathbf{\beta}_t}(\mathbf{q}_t))^2,
\end{eqnarray}
which implies that
\begin{eqnarray}
\phi_{\mathbf{\beta}_t}(\mathbf{b}_{t+1})\leq\phi_{\mathbf{\beta}_t}(\mathbf{b}_t)-\frac{3\lambda\eta^2\gamma^2}{2}.
\end{eqnarray}
\end{lemma}
\begin{proof}%[Proof of Theorem~\ref{thm:convergence}]
%Define function $\Phi$ by fixing $\mathbf{\beta}_t$ in function $\phi$, i.e., $\Phi(\mathbf{b}_t)=\phi(\mathbf{b}_t,\mathbf{\beta}_t)$.
%Hence, $\nabla\Phi(\mathbf{b}_t)=\nabla_{\mathbf{b}_t}\phi(\mathbf{b}_t,\mathbf{\beta}_t)$.

By Taylor's expansion at point $(\mathbf{b}_{t+1},\mathbf{\beta}_t)$ and the property of $\phi_{\mathbf{\beta}_t}$ in terms of $\lambda$,
\begin{eqnarray} \label{eq:taylor}
\phi_{\mathbf{\beta}_t}(\mathbf{b}_{t+1})\leq\phi_{\mathbf{\beta}_t}(\mathbf{b}_t)+\langle\nabla\phi_{\mathbf{\beta}_t}(\mathbf{b}_t),\mathbf{b}_{t+1}-\mathbf{b}_t\rangle+&\frac{\lambda}{2}\|\mathbf{b}_{t+1}-\mathbf{b}_t\|_2^2.
\end{eqnarray}
The last term in the above inequality is at most,
by $\mathbf{b}'=\mathbf{b}_t-\eta\nabla\phi_{\mathbf{\beta}_t}(\mathbf{b}_t)$ and the definition of projection,
\begin{eqnarray*}
&&\frac{\lambda}{2}\|\mathbf{b}_{t+1}-\mathbf{b}_t\|_2^2\leq\frac{1}{2}\lambda\eta^2\|\nabla\phi_{\mathbf{\beta}_t}(\mathbf{b}_t)\|_2^2
\leq\frac{1}{2}\lambda\eta^2\gamma^2.
\end{eqnarray*}

To bound the second term in (\ref{eq:taylor}), note that by the convexity of $\phi_{\mathbf{\beta}_t}$,
\[\phi_{\mathbf{\beta}_t}(\textbf{q}_t)\geq\phi_{\mathbf{\beta}_t}(\mathbf{b}_t)+\langle\nabla\phi_{\mathbf{\beta}_t}(\mathbf{b}_t),\mathbf{q}_t-\mathbf{b}_t\rangle,\]
which implies that
\[-\langle\nabla\phi_{\mathbf{\beta}_t}(\mathbf{b}_t),\mathbf{q}_t-\mathbf{b}_t\rangle\geq\phi_{\mathbf{\beta}_t}(\mathbf{b}_t)-\phi_{\mathbf{\beta}_t}(\mathbf{q}_t)\geq 0.\]
Let $\Delta=-\langle\nabla\phi_{\mathbf{\beta}_t}(\mathbf{b}_t),\mathbf{q}-\mathbf{b}_t\rangle$,
and we know from above that $\Delta^2\geq(\phi_{\mathbf{\beta}_t}(\mathbf{b}_t)-\phi_{\mathbf{\beta}_t}(\mathbf{q}))^2$.
Now we have the following lemma whose proof is in the next subsection.
\begin{lemma} \label{lem:convergence}
When $\Delta\geq 2\eta\gamma^2$,
\begin{eqnarray}
\langle\nabla\phi_{\mathbf{\beta}_t}(\mathbf{b}_t),\mathbf{b}_{t+1}-\mathbf{b}_t\rangle\leq-\frac{\eta}{3d^2}(\phi_{\mathbf{\beta}_t}(\mathbf{b}_t)-\phi_{\mathbf{\beta}_t}(\mathbf{q}_t))^2.
\end{eqnarray}
\end{lemma}

%By Lemma~\ref{lem:convergence}, we know that whenever $\phi_{\mathbf{\beta}_t}(\mathbf{b}_t)-\phi_{\mathbf{\beta}_t}(\mathbf{q})\geq d\gamma\sqrt{6\lambda\eta}$ for $q=\arg\min_{\mathbf{z}\in\mathcal{K}_t}\phi_{\mathbf{\beta}_t}(\mathbf{z})$,
%\begin{eqnarray*}
%\langle\nabla\phi_{\mathbf{\beta}_t}(\mathbf{b}_t),\mathbf{b}_{t+1}-\mathbf{b}_t\rangle&\leq&-\frac{\eta}{3d^2}(\phi_{\mathbf{\beta}_t}(\mathbf{b}_t)-\phi_{\mathbf{\beta}_t}(\mathbf{q}))^2\\
%&\leq&-2\lambda\eta^2\gamma^2
%\end{eqnarray*}
%where $\nabla\phi_{\mathbf{\beta}_t}(\mathbf{b}_t)=(\partial\phi_{\mathbf{\beta}_t}(\mathbf{b}_t)/\partial b_{ij})_{ij}$.

Suppose that $\phi_{\mathbf{\beta}_t}(\mathbf{b}_t)-\phi_{\mathbf{\beta}_t}(\mathbf{q}_t)\geq d\gamma\sqrt{6\lambda\eta}$, implying that $\Delta\geq d\gamma\sqrt{6\lambda\eta}$. By $\eta\leq\frac{\lambda d^2}{\gamma^2}$, i.e., $d\geq\gamma\sqrt{\frac{\eta}{\lambda}}$. This implies that
\[\Delta\geq d\gamma\sqrt{6\lambda\eta}\geq 2\eta\gamma^2,\]
which satisfies the condition of Lemma~\ref{lem:convergence}.
Therefore,
\begin{eqnarray*}
\phi_{\mathbf{\beta}_t}(\mathbf{b}_{t+1})&\leq&\phi_{\mathbf{\beta}_t}(\mathbf{b}_t)+\langle\nabla\phi_{\mathbf{\beta}_t}(\mathbf{b}_t),\mathbf{b}_{t+1}-\mathbf{b}_t\rangle+\frac{\lambda}{2}\|\mathbf{b}_{t+1}-\mathbf{b}_t\|_2^2\\
&\leq&\phi_{\mathbf{\beta}_t}(\mathbf{b}_t)-\frac{\eta}{3d^2}\Delta^2+\frac{\lambda\eta^2\gamma^2}{2}\\
&\leq&\phi_{\mathbf{\beta}_t}(\mathbf{b}_t)-\frac{\eta}{4d^2}\Delta^2,
\end{eqnarray*}
where the last inequality comes from $\Delta\geq d\gamma\sqrt{6\lambda\eta}$, i.e., $\lambda\eta\gamma^2\leq\frac{\Delta^2}{6d^2}$.
We obtain that
\begin{eqnarray*}
\phi_{\mathbf{\beta}_t}(\mathbf{b}_{t+1})\leq\phi_{\mathbf{\beta}_t}(\mathbf{b}_t)-\frac{\eta}{4d^2}(\phi_{\mathbf{\beta}_t}(\mathbf{b}_t)-\phi_{\mathbf{\beta}_t}(\mathbf{q}_t))^2,
\end{eqnarray*}
which, by $\phi_{\mathbf{\beta}_t}(\mathbf{b}_t)-\phi_{\mathbf{\beta}_t}(\mathbf{q}_t)\geq d\gamma\sqrt{6\lambda\eta}$, implies that
\begin{eqnarray*}
\Phi(\mathbf{b}_{t+1},{\mathbf{\beta}_t})\leq\Phi(\mathbf{b}_t,{\mathbf{\beta}_t})-\frac{3\lambda\eta^2\gamma^2}{2}.
\end{eqnarray*}
\end{proof}

Finally, the guarantee for Step~(2) of Algorithm~1 is not hard to see:
for each $i$, we know that $\beta_{t+1}(i)$ is the best possible value if we fix values of $b_{ij}$ to $b_{t+1}(i,j)$
so $\mathbf{\beta}_{t+1}$ never gives a larger objective value than $\mathbf{\beta}_t$.

\begin{lemma} \label{lem:opt}
By Step~(2) of Algorithm~1,
\begin{eqnarray} \label{eq:opt}
\Phi(\mathbf{b}_{t+1},\mathbf{\beta}_{t+1})\leq\Phi(\mathbf{b}_{t+1},\mathbf{\beta}_t).
\end{eqnarray}
\end{lemma}

The following main theorem gives the convergence rate based on the results of Lemma~\ref{thm:convergence} and Lemma~\ref{lem:opt}.
\begin{theorem} \label{thm:rate}
For $\eta=O(\frac{1}{t^{2/5}})$,
\begin{eqnarray}
\Phi(\mathbf{y}_t)-\Phi(\mathbf{y}^*)\leq\frac{2}{3\lambda\eta^2\gamma^2 t}=O(\frac{1}{\eta^2 t}).
\end{eqnarray}
\end{theorem}
\begin{proof}
By Lemma~\ref{thm:convergence} and Lemma~\ref{lem:opt} combined, we have that
\begin{eqnarray*}
\Phi(\mathbf{y}_{t+1})\leq\Phi(\mathbf{y}_t)-\frac{3\lambda\eta^2\gamma^2}{2},
\end{eqnarray*}
whenever $\phi_{\mathbf{\beta}_t}(\mathbf{b}_t)-\phi_{\mathbf{\beta}_t}(\mathbf{q}_t)\geq d\gamma\sqrt{6\lambda\eta}$ for $\mathbf{q}_t=\arg\min_{\mathbf{z}\in\mathcal{K}_t}\phi_{\mathbf{\beta}_t}(\mathbf{z})$,
from which we obtain the following claim by $\Phi(\cdot)\leq 1$:
\begin{eqnarray} \label{eq:recurrence}
\Phi(\mathbf{y}_{t+1})\leq\Phi(\mathbf{y}_t)-\frac{3\lambda\eta^2\gamma^2}{2}\Phi^2(\mathbf{y}_t).
\end{eqnarray}
Thus, we have that
\[\frac{1}{\Phi(\mathbf{y}_t)}-\frac{1}{\Phi(\mathbf{y}_{t-1})}\geq\frac{3\lambda\eta^2\gamma^2}{2}\cdot\frac{\Phi^2(\mathbf{y}_{t-1})}{\Phi(\mathbf{y}_{t-1})\Phi(\mathbf{y}_t)}\geq\frac{3\lambda\eta^2\gamma^2}{2},\]
where the first inequality is Inequality~(\ref{eq:recurrence}).
It can be concluded that 
\[\frac{1}{\Phi(\mathbf{y}_t)}\geq\frac{3\lambda\eta^2\gamma^2}{2}t,\]
from which along with the condition we have that
\[d\gamma\sqrt{6\lambda\eta}\leq\phi_{\mathbf{\beta}_t}(\mathbf{b}_t)-\phi_{\mathbf{\beta}_t}(\mathbf{q}_t)\leq\Phi(\mathbf{y}_t)-\Phi(\mathbf{y}^*)\leq\frac{2}{3\lambda\eta^2\gamma^2 t}.\]
Therefore, $\eta$ should be set as at most 
\[\frac{2}{3\sqrt{6}d\gamma^3\lambda^{3/2}t^{2/5}}=O(\frac{1}{t^{2/5}}).\]
\end{proof}

%\subsubsection{Convergence time} The value of the objective function has an upper bound, and the objective value decreases by at least $3\lambda\eta^2\gamma^2/2$ until the value of the objective function $\Phi$ reaches the set of $(\mathbf{b},\mathbf{\beta})$ that give values within $d\gamma\sqrt{6\lambda\eta}$ from the minimum. Thus, one can upper bound the convergence time.

\subsection{Proof of Lemma~\ref{lem:convergence}}
%By Taylor's expansion and the property of $\phi$ in terms of $\lambda$,
%\begin{eqnarray*}
%\phi(\mathbf{b}')\leq\phi(\mathbf{b})+\langle\nabla\phi(\mathbf{b}),\mathbf{b}'-\mathbf{b}\rangle+\frac{\lambda}{2}\|\mathbf{b}'-\mathbf{b}\|_2^2.
%\end{eqnarray*}
%The last term in (\ref{eq:taylor}) is at most, by $\mathbf{y}''=\mathbf{y}_t-\eta\nabla\phi(\mathbf{y}_t)$ and the definition of projection,
%\[\frac{\lambda}{2}\|\mathbf{y}_{t+1}-\mathbf{y}_t\|_2^2\leq\frac{1}{2}\lambda\eta^2\|\nabla\phi(\mathbf{y}_t)\|_2^2\leq\frac{1}{2}\lambda\eta^2\gamma^2.\]

%To bound the second term in (\ref{eq:taylor}), note that by the convexity of $\phi$,
%\[\phi(q)\geq\phi(\mathbf{b})+\langle\nabla\phi(\mathbf{b}),\mathbf{q}-\mathbf{b}\rangle,\]
%which implies that
%\[-\langle\nabla\phi(\mathbf{b}),\mathbf{q}-\mathbf{b}\rangle\geq\phi(\mathbf{b})-\phi(\mathbf{q})\geq 0.\]
%Let $\Delta=-\langle\nabla\phi(\mathbf{b}),\mathbf{q}-\mathbf{b}\rangle$,
%and we know from above that $\Delta^2\geq(\phi(\mathbf{b})-\phi(\mathbf{q}))^2$.
Now we are ready to show that when $\Delta\geq 2\eta\gamma^2$,
\begin{eqnarray} \label{eq:inner}
\langle\nabla\phi(\mathbf{b}),\mathbf{b}'-\mathbf{b}\rangle\leq-\frac{\eta\Delta^2}{3d^2}.
\end{eqnarray}
Note that $\phi(\mathbf{b})-\phi(\mathbf{q})\geq d\gamma\sqrt{6\lambda\eta}$ implies that $\Delta\geq d\gamma\sqrt{6\lambda\eta}$.
By $\eta\leq\frac{\lambda d^2}{\gamma^2}$, we have $d\sqrt{6\lambda}\geq 2\gamma\sqrt{\eta}$,
and thus $\Delta\geq d\gamma\sqrt{6\lambda\eta}\geq 2\eta\gamma^2$, which satisfies the condition for Inequality~(\ref{eq:inner}).

Consider the triangle formed by $\mathbf{b},\mathbf{b}',\bar{\mathbf{b}}$.
Let $\alpha_1,\alpha_2,\alpha_3$ denote the angles that correspond to the points $\mathbf{b},\mathbf{b}',\bar{\mathbf{b}}$, respectively.
Since $\mathbf{b}'=\Pi_{\mathcal{K}}(\bar{\mathbf{b}})$ and $\mathbf{b}\in\mathcal{K}$,
we know that $\alpha_2\geq\frac{\pi}{2}$ and thus \textcolor[rgb]{1.00,0.00,0.00}{$\alpha_1\leq\frac{\pi}{2}$}.
(When $\mathbf{y}'=\mathbf{y}''$, we use the convention that $\alpha_2=\alpha_3=\frac{\pi}{2}$.)
As $\bar{\mathbf{b}}=\mathbf{b}-\eta\nabla\phi(\mathbf{b})$, we have that
\begin{eqnarray*}
\langle\nabla\phi(\mathbf{b}),\mathbf{b}'-\mathbf{b}\rangle&=&\|\nabla\phi(\mathbf{b})\|_2\|\mathbf{b}'-\mathbf{b}\|_2\cos(\pi-\alpha_1)\\
&=&-\|\nabla\phi(\mathbf{b})\|_2\|\mathbf{b}'-\mathbf{b}\|_2\cos\alpha_1,
\end{eqnarray*}
which is at most 0 since $\cos\alpha_1\geq 0$. Moreover, because $\alpha_2\geq\frac{\pi}{2}$, we have that
\[\|\mathbf{b}'-\mathbf{b}\|_2\geq\|\mathbf{b}-\bar{\mathbf{b}}\|_2\sin\alpha_3=\eta\|\phi(\mathbf{b})\|_2\sin\alpha_3.\]
Thus, we derive that
\[\langle\nabla\phi(\mathbf{b}),\mathbf{b}'-\mathbf{b}\rangle\leq-\eta\|\phi(\mathbf{b})\|_2^2\cos\alpha_1\sin\alpha_3.\]
What is left to be shown are the lower bounds for $\cos\alpha_1$ and $\sin\alpha_3$.

Consider the triangle formed by $\mathbf{b}, \mathbf{q}, \mathbf{b}'$, and let $\alpha_4$ denote the angle at point $\mathbf{b}$.
Then, we need the following two claims, whose proofs will be provided in the following two subsections.
\begin{lemma} \label{lm:1}
If $\Delta\geq 2\eta\gamma^2$, $\alpha_4\geq\alpha_1$.
\end{lemma}
\begin{lemma} \label{lm:2}
If $\Delta\geq 2\eta\gamma^2$, $\sin\alpha_3\geq\frac{\cos\alpha_4}{3}$.
\end{lemma}

With these two claims, we obtain
\[\langle\nabla\phi(\mathbf{b},\mathbf{b}'-\mathbf{b})\rangle\leq-\frac{\eta\|\nabla\phi(\mathbf{b})\|_2^2\cos^2\alpha_4}{3}.\]
Recall that $-\Delta=\langle\nabla\phi(\mathbf{b}),\mathbf{q}-\mathbf{b}\rangle=
\|\nabla\phi(\mathbf{b})\|_2\|\mathbf{q}-\mathbf{b}\|_2\cos(\pi-\alpha_4)=
-\|\nabla\phi(\mathbf{b})\|_2\|\mathbf{q}-\mathbf{b}\|_2\cos\alpha_4$.
Thus,
\[\cos\alpha_4=\frac{\Delta}{\|\nabla\phi(\mathbf{b})\|_2\|\mathbf{q}-\mathbf{b}\|_2}\geq\frac{\Delta}{d\|\nabla\phi(\mathbf{b})\|_2},\]
and we have proven Inequality~(\ref{eq:inner}).

\subsubsection{Proof of Lemma~\ref{lm:1}}
Consider the triangle formed by $\mathbf{b},\bar{\mathbf{b}},\mathbf{q}$.
Let $\bar{\mathbf{b}}'$ be the point on plane $\mathbf{b}-\bar{\mathbf{b}}-\mathbf{q}$ such that the triangles $\mathbf{b}-\bar{\mathbf{b}}-\bar{\mathbf{y}}'$ and $\mathbf{b}-\bar{\mathbf{b}}-\mathbf{b}'$ are identical. ($\bar{\mathbf{b}}'=\mathbf{b}'$ if the triangles $\mathbf{b}-\bar{\mathbf{b}}-\mathbf{q}$ and $\mathbf{b}-\bar{\mathbf{b}}-\mathbf{b}'$ are on the same plane.)
Clearly, angle $\bar{\mathbf{b}}-\mathbf{b}-\bar{\mathbf{b}}'$ equals angle $\bar{\mathbf{b}}-\mathbf{b}-\mathbf{b}'$, which is $\alpha_1$.
Thus, we can consider triangle $\mathbf{b}-\bar{\mathbf{b}}-\bar{\mathbf{b}}'$ instead of $\mathbf{b}-\bar{\mathbf{b}}-\mathbf{b}'$.

Assume for the sake of contradiction that $\alpha_4 < \alpha_1$.
Note that since $\eta\leq\frac{\Delta}{2\gamma^2}\leq\frac{\Delta}{\gamma^2}$,
\begin{eqnarray*}
\|\mathbf{b}-\bar{\mathbf{b}}'\}_2&\leq&\|\mathbf{b}-\bar{\mathbf{b}}\|_2=\eta\|\nabla\phi(\mathbf{b})\|_2\leq\frac{\Delta}{\gamma}\\
&\leq&\frac{|\langle\nabla\phi(\mathbf{b}),\mathbf{b}-\mathbf{q}\rangle|}{\|\nabla\phi(\mathbf{b})\|_2}\leq\|\mathbf{b}-\mathbf{q}\|_2.
\end{eqnarray*}
This implies that line $\mathbf{b}-\mathbf{q}$ must intersect the line $\bar{\mathbf{b}}'-\bar{\mathbf{b}}$ at some point $\mathbf{z}$ with $\|\mathbf{z}-\bar{\mathbf{b}}\|_2<\|\bar{\mathbf{b}}'-\bar{\mathbf{b}}\|_2$.
However, since $\mathbf{b},\mathbf{q}\in\mathcal{K}$ and $\mathcal{K}$ is convex, we know that $\mathbf{z}\in\mathcal{K}$,
we know that $\mathbf{z}\in\mathcal{K}$ and thus $\|\mathbf{z}-\bar{\mathbf{b}}\|_2\geq\|\bar{\mathbf{b}}'-\bar{\mathbf{b}}\|_2$ by the definition of $\mathbf{b}'$.
This is a contradiction as $\|\mathbf{b}'-\bar{\mathbf{b}}\|_2=\|\bar{\mathbf{b}}'-\bar{\mathbf{b}}\|_2$.
We can conclude that $\alpha_4\geq\alpha_1$.

\subsubsection{Proof of Lemma~\ref{lm:2}}
Since $\alpha_4\geq\alpha_1$ and $\|\mathbf{y}_t-\mathbf{q}\|_2\geq\|\mathbf{y}_t-\bar{\mathbf{y}}'\|_2$,
$\bar{\mathbf{b}}'$ must lie inside triangle $\mathbf{b}-\bar{\mathbf{b}}-\mathbf{q}$,
which implies that the line passing through $\mathbf{q}$ and $\bar{\mathbf{b}}'$ must intersect line segment $\mathbf{b}-\bar{\mathbf{b}}$ at some point $\mathbf{w}$.
Let $a_1,a_2,a_3$ denote angles $\mathbf{b}-\mathbf{q}-\bar{\mathbf{b}}'$, $\mathbf{q}-\bar{\mathbf{b}}'-\bar{\mathbf{b}}$, and $\mathbf{w}-\bar{\mathbf{b}}-\bar{\mathbf{b}}'$, respectively.

We first claim that $a_2\geq\frac{\pi}{2}$. To see that, we compare it with angle $\bar{\mathbf{b}}-\mathbf{b}'-\mathbf{q}$, denoted as $a'_2$,
which is at least $\frac{\pi}{2}$.
Since triangles $\mathbf{b}-\bar{\mathbf{b}}-\bar{\mathbf{b}}'$ and $\mathbf{b}-\bar{\mathbf{b}}-\mathbf{b}'$ are identical,
angle $\mathbf{b}-\bar{\mathbf{b}}-\bar{\mathbf{b}}'$ (denoted as $\bar{\theta}$) is the same as angle $\mathbf{b}-\bar{\mathbf{b}}-\mathbf{b}'$ (denoted as $\theta$), which implies that angle $\bar{\mathbf{b}}'-\bar{\mathbf{b}}-\mathbf{q}$ is at most angle $\mathbf{b}'-\mathbf{b}-\mathbf{q}$ since $\bar{\mathbf{b}}'$ lies on the same plane as $\mathbf{b}-\bar{\mathbf{b}}-\mathbf{q}$.
Since $\|\mathbf{b}'-\bar{\mathbf{b}}\|_2=\|\bar{\mathbf{b}}'-\bar{\mathbf{b}}\|_2$,
if we rotate triangles $\mathbf{b}'-\bar{\mathbf{b}}-\mathbf{q}$ and $\bar{\mathbf{b}}'-\bar{\mathbf{b}}-\mathbf{q}$ along line $\mathbf{q}-\bar{\mathbf{b}}$ to make them lie on the same plane, $\bar{\mathbf{b}}'$ must lie inside of triangle $\mathbf{b}'-\bar{\mathbf{b}}-\mathbf{q}$.
As $a'_2\geq\frac{\pi}{2}$, we must have that $a_2\geq a'_2\geq\frac{\pi}{2}$.

Next, note that $\alpha_3=a_2-a_3\geq\frac{\pi}{2}-a_3$ and $a_3=a_1+\alpha_4$.
So, $\sin\alpha_3\geq\sin(\frac{\pi}{2}-a_1-\alpha_4)=\cos(a_1+\alpha_4)=\cos a_1\cos\alpha_4-\sin a_1\sin\alpha_4\geq\cos\alpha_4\sqrt{1-\sin^2 a_1}-\sin a_1$ as $\sin\alpha_4\leq 1$ and $\sin a_1\geq 0$. Moreover,
\begin{eqnarray*}
\sin a_1&\leq&\frac{\|\mathbf{b}-\mathbf{w}\|_2}{\|\mathbf{b}-\mathbf{q}\|_2}\leq\frac{\|\mathbf{b}-\bar{\mathbf{b}}\|_2}{\mathbf{b}-\mathbf{q}}\\
&\leq&\frac{\eta\|\nabla\phi(\mathbf{b})\|_2}{(\frac{\Delta}{\|\nabla\phi(\mathbf{b})\|_2\cos\alpha_4})}
\leq\frac{\eta\gamma^2\cos\alpha_4}{\Delta}\leq\frac{\cos\alpha_4}{2}.
\end{eqnarray*}
As a result, we have that $\sin\alpha_4\geq\cos\alpha_4\sqrt{1-\frac{1}{4}}\geq\frac{\cos\alpha_4}{3}$.

\section{Future Work}
Our current analysis can only provide a lower convergence rate than the BCGP and two-block alternating minimization algorithms in \cite{beck:tetruashvili}. Deriving a better convergence rate is an immediate future work. 
For distributed computation, although the projection is applied to the whole vector $\mathbf{b}$ as in \cite{birnbaum}, alternatively, one may apply mirror descents with respect to $b_i$ separately, and 
show that, jointly, the objective still converges to minimum (as has been done for congestion games
\cite{CL14,CL15,CL16}).

Agents may strategically report their interests in different services in order to maximize their own utilities in markets.
When we allow this and treat the market equilibrium computation as a mapping from a given report, we are studying market games.
As far as we know, it is still an open question  Nash equilibria exist in ``Arrow-Debreu market games" (as in Fisher market games \cite{branzei}).
It would be also interesting to consider different classes of underlying social networks that describe crowd structures \cite{kakade}.
There are some price-of-anarchy results in Fisher market games \cite{branzei}.
Another interesting research direction might be better bounding the price of anarchy for Fisher/Arrow-Debreu market games.

\subsubsection*{Acknowledgements.}
We would like to thank Ling-Wei Wang for useful discussions.
%%%%%%%%%%%%%%%%%%%%%%%%%%%%%%%%%%%%%%%%%%%%%%%%%%%%%%%%%%%%%%%%%%%%%%%%%%%%%%%%%%%%%%%%%%%%%%%%%%%%%%%%%
%% bibliography: see CFP for number of permitted pages

\bibliographystyle{plain}  % do not change this line!
%\bibliography{markets}

\appendix
\section{Argument of Condition (*) for Theorem~\ref{thm:existence}}
We restate the argument of \cite{devanur:garg} here.
Assume that $\{k\}$ is a singleton strongly connected component without a loop.
Let $T$ denote the set of nodes different from $\{k\}$ that can be reached on a directed path in $E$ from $\{k\}$.
In an equilibrium allocation, the agents in $T \cup \{k\}$ spend all their money on the
goods of the agents in $T$, which implies that $p_k = 0$, contrary to our assumption that $p_j>0$ for every $j\in A$.

\section{Proof of Proposition~\ref{pro:hessian}}
Let $\mathbf{b} \in S_\mathbf{\beta}$. Consider any $i,k \in A$, $j,l \in G$. First, we have
\begin{equation}
\nonumber \frac{\partial \phi_\mathbf{\beta}(\mathbf{b})}{\partial b_{ij}} = -\log u_{ij} + \log \frac{p_j}{\beta_i} +1= 1-(\log u_{ij} - \log \frac{p_j}{\beta_i}) = 1-\log \frac{u_{ij}\beta_i}{p_j}.
\end{equation}
Let $p_j=\sum_i b_{ij}$. Thus, we have
\begin{equation}
\nonumber \frac{\partial \phi_\mathbf{\beta}(\mathbf{b})}{\partial b_{ij}} = 1-\log \frac{u_{ij}\beta_i}{\sum_i b_{ij}}.
\end{equation}
Next, note that except for $i = k$, we have
\begin{equation}
\nonumber  \frac{\partial^2 \phi_\mathbf{\beta}(\mathbf{b})}{\partial b_{ij} \partial b_{kl}}=0 ,
\end{equation}
and if $i = k$, we have
\begin{equation}
\nonumber \frac{\partial^2 \phi_\mathbf{\beta}(\mathbf{b})}{\partial b_{ij} \partial a_{kl}} = \frac{1}{\sum_i b_{ij}}=\frac{1}{p_j} .
\end{equation}
This means that each entry of the Hessian matrix $\nabla^2 \phi_\mathbf{\beta}(\mathbf{b})$ is at most $n$. Then for any $z \in \mathbb{R}$, we have
\begin{equation}
\begin{aligned}
\nonumber z^\top\cdot (\nabla^2 \phi_\mathbf{\beta}(b))\cdot z &= (\sum_i z_{ij} \frac{1}{p_j})_{ij} \cdot z \\
						                                                      &= \sum_{i,j} z_{ij} (\sum_i z_{ij} \frac{1}{p_j}) \\
						                                                      &= \sum_j \frac{1}{p_j} (\sum_i z_{ij})^2 \\
						                                                      & \leq \sum_j (z_{1j}+z_{2j}+...+z_{nj})^2 \\
						                                                      &\leq \sum_j (|z_{1j}|+|z_{2j}|+...+|z_{nj}|)^2 \\
						                                                      &\leq \sum_j (1^2+1^2+...+1^2)(z_{1j}^2+z_{2j}^2+...+z_{nj}^2) \\
						                                                      &= n ( \sum_{i,j} z^2_{ij}).
\end{aligned}
\end{equation}
The first inequality is by $p_i \geq 1$. The third inequality is by Cauchy-Schwarz inequality. This implies that $\nabla^2 \phi_\mathbf{\beta}(\mathbf{b}) \preceq \alpha I$ with $\alpha=n$.

\end{document}